\documentclass[11pt]{article}

\usepackage[letterpaper,margin=1.00in]{geometry}
\usepackage{amsmath, amssymb, amsthm, amsfonts}
\usepackage{bbm}

\usepackage{ifthen}

\usepackage{appendix}
\usepackage{graphicx}
\usepackage{color}
\usepackage{algorithm}
\usepackage[noend]{algpseudocode}
\usepackage{wrapfig}
\usepackage{balance}

\usepackage{framed}
\usepackage[framemethod=tikz]{mdframed}
\usepackage[bottom]{footmisc}
\usepackage{enumitem}

\let\oldtextbf=\textbf
\renewcommand\textbf[1]{{\boldmath\oldtextbf{#1}}}

\definecolor{darkgreen}{rgb}{0,0.5,0}
\usepackage{hyperref}
\usepackage{thmtools}
\usepackage{thm-restate}

\hypersetup{
    unicode=false,          
    colorlinks=true,        
    linkcolor=red,          
    citecolor=darkgreen,        
    filecolor=magenta,      
    urlcolor=cyan           
}
\usepackage[capitalize, nameinlink]{cleveref}
\newtheorem{theorem}{Theorem}[section]
\newtheorem{lemma}[theorem]{Lemma}
\newtheorem{meta-theorem}[theorem]{Meta-Theorem}

\newtheorem{observation}[theorem]{Observation}
\newtheorem{definition}[theorem]{Definition}
\Crefname{remark}{Remark}{Remarks}
\Crefname{observation}{Observation}{Observations}
\algnewcommand\algorithmicswitch{\textbf{switch}}
\algnewcommand\algorithmiccase{\textbf{case}}

\algdef{SE}[SWITCH]{Switch}{EndSwitch}[1]{\algorithmicswitch\ #1\ \algorithmicdo}{\algorithmicend\ \algorithmicswitch}%
\algdef{SE}[CASE]{Case}{EndCase}[1]{\algorithmiccase\ #1}{\algorithmicend\ \algorithmiccase}%
\algtext*{EndSwitch}%
\algtext*{EndCase}%

\newcommand{\poly}{\operatorname{\text{{\rm poly}}}}

\newcommand{\codots}{,\cdots,}

\renewcommand{\paragraph}[1]{\vspace{0.15cm}\noindent {\bf #1}.}

\newcommand{\FullOrShort}{full}

\ifthenelse{\equal{\FullOrShort}{full}}{
    
  \newcommand{\fullOnly}[1]{#1}
  \newcommand{\shortOnly}[1]{}
  }{
    \newcommand{\fullOnly}[1]{}
    \newcommand{\shortOnly}[1]{#1}
  }

\begin{document}
\title{Nearly Work-Efficient Parallel DFS in Undirected Graphs}
\date{}

\author{Mohsen Ghaffari \\ MIT \\ ghaffari@mit.edu \and  Christoph Grunau \\ ETH Zurich \\ cgrunau@ethz.ch \and Jiahao Qu \\ ETH Zurich \\ jiahao.qu@inf.ethz.ch}




\maketitle

\begin{abstract} 
We present the first parallel depth-first search algorithm for undirected graphs that has near-linear work and sublinear depth. Concretely, in any $n$-node $m$-edge undirected graph, our algorithm computes a DFS in $\tilde{O}(\sqrt{n})$ depth and using $\tilde{O}(m+n)$ work. All prior work either required $\Omega(n)$ depth, and thus were essentially sequential, or needed a high $\poly(n)$ work and thus were far from being work-efficient. 
\end{abstract}


\section{Introduction}
Depth-first search (DFS) is one of the basic algorithmic techniques for graph problems, with a wide range of applications. In an $n$-node and $m$-edge undirected graph, a simple sequential algorithm computes a DFS in $O(m+n)$ time. This is often covered in introductory algorithmic courses.  Unfortunately, the state-of-the-art parallel algorithms for computing a DFS require at least $\Omega(\sqrt{n})$ processors to run faster than this sequential algorithm. This significantly limits the applicability of these parallel DFS algorithms, because having $\Omega(\sqrt{n})$ or more processors is quite a high requirement for most plausible applications. In this paper, we describe a parallel DFS algorithm that runs faster than the sequential algorithm with just $\poly(\log n)$ processors. Indeed, so long as the number of available processors is in the range $[1, \Theta(\sqrt{n})]$---which arguably captures a wide range of practical settings of interest---this parallel DFS algorithm provides the best possible speedup over the sequential algorithm, up to logarithmic factors.

We next overview notions of work and depth in parallel algorithms and how they determine the time complexity given a number of processors. Then, we review prior work on parallel DFS computations. Afterward, we formally state our contributions. 

\subsection{Background: Work \& depth}
\noindent \textbf{Work and depth in parallel algorithms.} We follow the standard work-depth terminology~\cite{blelloch1996programming}.\footnote{More concretely, we describe our work assuming the strongest PRAM variant, CRCW with arbitrary writes. This is done for simplicity. The results can be extended easily to the weaker variants, e.g., EREW, as the latter can simulate the stronger variants at the cost of a logarithmic factor loss in depth and work. We did not attempt to optimize the logarithmic factors in our results.} For an algorithm $\mathcal{A}$, the \textit{work} $W(\mathcal{A})$ is defined as the total number of operations. The \textit{depth} $D(\mathcal{A})$ is defined as the longest chain of operations with sequential dependencies, in the sense that the $(i+1)^{th}$ operation depends on (and should wait for) the results of operations $i$ in the chain. 
The work and depth bounds determine the time $T_p(\mathcal{A})$ for running the algorithm when we have $p$ processor. We have $T_{1}(\mathcal{A}) = W(\mathcal{A})$ and $T_{\infty}(\mathcal{A}) = D(\mathcal{A}).$ A simple observation, known as Brent's principle~\cite{brent1974parallel}, gives the following general bound: $$W(\mathcal{A})/p \leq T_{p}(\mathcal{A}) \leq W(\mathcal{A})/p + D(\mathcal{A}).$$

\paragraph{Work-efficient parallel algorithms} Parallel algorithms with work $W(\mathcal{A})$ asymptotically equal to the complexity of their sequential counterpart are known as \textit{work-efficient}. If this equality holds up to a $\poly(\log n)$ factor, the algorithm is called \textit{nearly work-efficient}. (Nearly) work-efficient parallel algorithms enjoy asymptotically optimal speed-up over sequential algorithms  (up to logarithmic factors) for a small number of processors. Once the number of processors exceeds some threshold ---asymptotically equal to $W(\mathcal{A})/D(\mathcal{A})$---the time complexity bottoms out at $D(\mathcal{A})$. Thus, an ultimate goal in devising parallel algorithms is to obtain (nearly) work-efficient algorithms with depth as small as possible.

\subsection{State of the art for parallel DFS algorithms} DFS is quite hard for parallel algorithms. Reif~\cite{reif1985depth} showed that computing the lexicographically first DFS---where the DFS should visit the neighbors of a node according to their numbers---is $\mathsf{P}$-complete.\footnote{Technically, the $\mathsf{P}$-completeness is for a decision variant which asks whether a vertex $v$ is visited before another vertex $u$ or not.} That is, if there is a $\poly(\log n)$-depth $\poly(n)$ work algorithm for this problem, then there is such an algorithm for all problems in $\mathsf{P}$. This is why his paper was titled ``\textit{Depth-first search is inherently sequential.}'' Follow-up work on parallel DFS algorithms thus focused on computing an arbitrary DFS, where the order of visiting neighbors of a node gets chosen by the algorithm. This is also the version of the DFS problem that we tackle in this paper.

The most relevant prior work is by Aggarwal and Anderson~\cite{Aggarwal87} and a follow-up by Goldberg, Plotkin, and Vaidya~\cite{goldberg1988sublinear}. These give parallel DFS algorithms, though focusing almost exclusively on depth and require work much higher than the sequential algorithm. 

Aggarwal and Anderson~\cite{Aggarwal87}, building partially on a prior work of Anderson\cite{anderson1985parallel}, presented a randomized DFS algorithm for undirected graphs, with $\poly(\log n)$ depth and $\poly(n)$ work.\footnote{Aggarwal, Anderson, and Kao~\cite{aggarwal1989parallel} later provided an extension to directed graphs, with similar bounds. In this paper, our focus is on undirected graphs.} The latter is a high and unspecified polynomial in $n$, which is at least $\Omega(n^3)$. This high work complexity is, in part, due to its use of exact maximum weight matching as a subroutine, for which known parallel algorithms need a high $\poly(n)$ work~\cite{karp1985constructing}. Goldberg et al.\cite{goldberg1988sublinear} devised a deterministic variant of \cite{Aggarwal87} with depth $\tilde{O}(\sqrt{n})$ and work $\poly(n)$. A closer inspection indicates that this deterministic algorithm requires $\tilde{\Omega}(m\sqrt{n})$ work. 

To summarize, the state-of-the-art parallel DFS algorithms take work much higher work than the sequential algorithm. Thus, they require a high number of processors to run faster than the sequential algorithm. Even in the work of Goldberg et al.\cite{goldberg1988sublinear}, one would need $\Omega(\sqrt{n}\cdot \poly(\log n))$ processors to see any time advantage over the sequential algorithm. This issue significantly limits the applicability of these algorithms~\cite{goldberg1988sublinear, Aggarwal87} in the current (or plausible future) settings of parallel computation.

\subsection{Our contribution}
We present the first nearly work-efficient parallel algorithm with sublinear depth. This algorithm will outperform the sequential counterpart as soon as we have $\poly(\log n)$ processors. In fact, it exhibits an optimal speedup over the sequential algorithm, up to logarithmic factors, if the number of processors is at most $\Theta(\sqrt{n})$. Arguably, this range covers most of the practically relevant settings.

\begin{theorem}\label{thm:main} There is a randomized parallel algorithm that, in any $n$-node $m$-edge undirected graph $G=(V, E)$, given a root node $r$, computes a depth-first search tree of $G$ rooted at node $r$, using $\tilde{O}(m+n)$ work and $\tilde{O}(\sqrt{n})$ depth, with high probability.
\end{theorem}
\fullOnly{
We also sketch in \Cref{app:deterministic} how we can achieve the same statement as \Cref{thm:main} using a deterministic algorithm.}\shortOnly{In the full version of this paper, we sketch how we can achieve the same statement as \Cref{thm:main} using a deterministic algorithm.} This is by replacing some randomized subroutines in algorithms that we use from prior work with deterministic counterparts, and it increases the depth and work bounds by only logarithmic factors.

Our algorithm for \Cref{thm:main} follows the outer shell of the approach of Aggarwal and Anderson~\cite{Aggarwal87}. The novelty is in the internal ingredients. We adjust some parts of the algorithm to make it nearly work-efficient, and in particular, we use (and develop) certain batch-dynamic parallel data structures in several parts. The latter allows us to ensure that the total work remains $\tilde{O}(m+n)$. We provide an algorithm overview in \Cref{sec:overview}.

\section{Preliminaries}
\subsection{Basic definitions}
\label{subsec:definitions}
We use two basic definitions from prior work~\cite{Aggarwal87}:
\begin{definition} (\textbf{Initial DFS segment})
    Consider a graph $G=(V, E)$ and a root node $r\in V$. An \textbf{initial DFS segment}, or simply an initial segment, is a tree $T'=(V',E')$ rooted in node $r$, where $V'\subseteq V$ and $E'\subseteq E$, such that $T'$ can be extended to some depth-first search tree $T''$ rooted in node $r$. That is, there exists a depth-first search tree $T''=(V, E'')$ such that $E'' \supseteq E'$. 
\end{definition}

\begin{observation}\label{obs:UniqueMin} A rooted tree $T'$ is an initial segment iff there are no paths between different branches of $T'$ using vertices in $V-T'$. More formally, there should be no path connecting two incomparable nodes of $T'$ and made of internal nodes in $V-T'$. Here, incomparable means two nodes of $T'$, neither of which is an ancestor of the other.
\end{observation}

\begin{definition} (\textbf{Separator})
    For a graph $H$ with $n'$ vertices, a subset of vertices $Q$ is called a \textbf{separator} if the largest connected component of $H - Q$ has size at most $\frac{n'}{2}$. We call a separator $Q$ a \textbf{$k$-path separator} if $Q$ is made of $k$ vertex disjoint paths.
\end{definition} 

\subsection{Basic tools from prior work}
\noindent\textbf{Prefix sum on a linked list.} The prefix sums in an $n$-item linked list can be computed in $O(\log n)$ depth and $O(n)$ work~\cite{ANDERSON1990269}:   

\begin{lemma}
\label{lemma:andersonmiller}
    Given a linked list $(x_1\codots x_k)$ where each element $x_i$ is associated with a number $y_i$, there is a deterministic algorithm that computes the prefix sum in $O(\log n)$ depth and $O(n)$ work. Moreover, the value $\sum_{j=1}^i y_j$ can be accessed directly from $x_i$.
\end{lemma}

\paragraph{Maximal matching} Luby~\cite{LUBY1993250} provides an efficient deterministic maximal matching algorithm with $O(\log^5 n)$ depth and $\tilde{O}(m)$ work. We will use this in a black-box manner. 
\begin{lemma}
\label{lem:lubylemma}
Given a graph $G=(V,E)$, there is a deterministic algorithm that computes a maximal matching in $O(\log^5 n)$ depth and $O(m \log^5 n)$ work. 
\end{lemma}

\section{Algorithm Overview}
\label{sec:overview}The outer shell of our algorithm is based on the classic approach of Aggarwal and Anderson\cite{Aggarwal87}, which provides a $\poly(\log n)$-depth DFS algorithm but uses $\poly(n)$ work for a high polynomial. The approach is recursive. Let $G$ be the input graph, and $r$ be the DFS root. We gradually grow an \textit{initial DFS segment} of $G$ until it becomes a complete DFS tree of $G$. See \Cref{subsec:definitions} for definitions. At the start, the segment is simply the root $r$. At each point, the initial segment is extended such that the problem is reduced to finding a new DFS tree in each connected component of the remaining graph, and such that each component has size at most half of the previous size. Thus, within $\log n$ recursions, the whole DFS is constructed. 

The algorithm for extending an initial segment $T'$ to a full DFS tree works as follows. For a connected component $C$ in $G-T'$, by \Cref{obs:UniqueMin}, there is a unique vertex $x \in T'$ with the lowest depth that has a neighbor $y\in C$. We construct a DFS rooted in $y$ for component $C$, then connect it to $T'$ using the edge $(x,y)$. This is done in parallel for different connected components of $G-T'$. 

The core of the algorithm is to construct an initial DFS segment that forms a separator for each component. This involves two parts: The first part is to find a separator that consists of a small number of vertex disjoint paths. The second part is constructing an initial segment from this set of paths. We next discuss each of these parts separately, in \Cref{subsec:separator} and \Cref{subsec:initialSegment}, and comment on how our algorithm differs from that of Aggarwal and Anderson and achieves $\tilde{O}(m)$ work.\footnote{From now on, we focus on connected graphs. This implies that $m=\Omega(n)$ and thus allows us to state the work bound simply as $\tilde{O}(m)$, instead of $\tilde{O}(m+n)$. Note that connected components can be identified in $\tilde{O}(m)$ work and $\poly(\log n)$ depth via classical parallel algorithms~\cite{JaJ92}. } 


\subsection{Separator construction} 
\label{subsec:separator}
The first part is to construct a separator with few paths. 

\paragraph{Aggarwal-Anderson} Aggarwal and Anderson~\cite{Aggarwal87} present a $\poly(n)$-work and $\poly(\log n)$ depth algorithm that computes an $O(1)$-path separator, as follows: We start with the trivial $n$-path separator that consists of one path for each vertex. Then, in each iteration, we reduce the number of paths by a constant factor while ensuring that the paths form a separator. We continue this until at most a constant number of paths are left. The core of the process is reducing the number of paths by a constant factor. To do that, in a rough sense, the basic idea is to match up and merge pairs of paths iteratively. Aggarwal and Anderson~\cite{Aggarwal87} reduced this problem to a minimum weight perfect matching problem. The latter is known to be solvable using $\poly(\log n)$ depth and a high $\poly(n)$ amount of work\cite{karp1985constructing}, which is at least $\Omega (n^3)$ work. However, this approach does not yield a work-efficient parallel algorithm as we are unaware of any $\tilde{O}(m)$-work algorithm for minimum-weight perfect matching, even with a sublinear depth. 

\paragraph{Our separator algorithm} To keep the work bound $\tilde{O}(m)$, our separator will consist of $\tilde{O}(\sqrt{n})$ paths, instead of $O(1)$. In \cref{separatorconstruction}, we present a parallel algorithm that computes an $O(\sqrt{n})$-path separator using near-linear work and $\tilde{O}(\sqrt{n})$ depth: 

\begin{restatable}[\textbf{Separator Theorem}]{theorem}{separator}
\label{thm:separator}
   There is an algorithm that finds an $O(\sqrt{n})$-path separator $Q$ in $O(\sqrt{n} \log^8 n)$ depth and $O(m \log^7 n)$ work. Each path is stored as one doubly-linked list. 
\end{restatable}

\subsection{Construting an initial segment from the separator paths}
\label{subsec:initialSegment}
The second part assumes that we have a separator consisting of several paths, and absorbs these paths into the current partial DFS tree, forming a new initial DFS segment that includes all vertices of these paths (and potentially some more).

\paragraph{Aggarwal-Anderson} Aggarwal and Anderson find a separator $Q$ that consists of only $O(1)$ paths. Then, they add these paths to the partial DFS tree essentially one by one. They find a path $p$ from the lowest node in the partial tree $T'$ to a path $l$ in the separator. Then, by adding the path $p$ to the partial DFS tree $T'$, they can absorb at least half of the vertices of $l$ into $T'$. These operations can be done using basic parallel spanning tree algorithms, in $\tilde{O}(m)$ depth and $\poly(\log n)$ work. After repeating the above procedure at most $O(\log n)$ times, all vertices in the separator get absorbed into $T'$.


\paragraph{Our absorption algorithm} 
In contrast to the $O(1)$-path separator of Aggarwal and Anderson, our work-efficient separator construction produces $O(\sqrt{n})$ paths. If we were to trivially absorb these paths, the work required with the basic solution would be $\tilde{O}(m \sqrt{n})$. To make this part work-efficient, we need that the total work over all absorptions is $\tilde{O}(m)$. We will perform each absorption using depth $\poly(\log n)$ and work linearly proportional to the sum of the number of vertices on the paths and the number of edges adjacent to the vertices on the path. These vertices and edges get deleted from the remaining graph due to the absorption, and hence we can argue that the overall work is $\tilde{O}(m)$. The key algorithmic ingredient will be certain batch-dynamic parallel data structures, which can perform large batches of updates to the graph using depth $\poly(\log n)$ and work proportional to the total number of changes in the graph. In \Cref{sec:initialsegment} and \Cref{sec:dataStructure}, we present the algorithms that provide this and prove the following theorem statement.

\begin{restatable}[\textbf{Absorption Theorem}]{theorem}{merging}
\label{thm:merging}
    Given an $O(\sqrt{n})$-path separator, where each path is stored as a doubly-linked list, there is an algorithm that constructs an initial segment $T' \subseteq G$, where $T'$ is also a separator for $G$. The algorithm uses $O(\sqrt{n} \log^3 n)$ depth w.h.p. and $O(m \log^3 n)$ work in expectation.
\end{restatable}


\section{Separator construction}
\label{separatorconstruction}

In this section, we prove the following theorem.
\separator*
\begin{proof}
    We start with $Q$ consisting of one path for each vertex. By \cref{Lemma reduce path}, when $Q$ consists of more than $48 \sqrt{n}$ paths, we can repeatedly reduce the number of paths in $Q$ by a constant factor in $O(\sqrt{n} \log^7 n)$ depth and $O(m \log^6 n)$ work. After $O(\log n)$ iterations, we get a separator $Q$ consisting of at most $48\sqrt{n}$ paths.
\end{proof}

\subsection{Path Reduction}

\subsubsection{Review of Path Reductions in \cite{Aggarwal87}}
The crucial part of the separator construction is to reduce the number of paths by a constant fraction while preserving the separator property of $Q$. 

Suppose that initially, $Q$ has $k$ paths. Aggarwal and Anderson~\cite{Aggarwal87} first divide the paths into a set $L$ of long paths and a set $S$ of short paths. Initially, $\frac{1}{4}k$ of the paths are placed in $L$ and the rest $\frac{3}{4}k$ are placed in $S$. The idea is to find a set of vertex disjoint paths $P$ between $L$ and $S$. Each path $p$ in $P$ has one end on a long path and the other end on a short path. All internal vertices of $p$ are not contained in $Q$, and each path in $Q$ intersects at most one path in $P$. We refer to a set of paths $P$ that satisfy all properties stated above as \emph{valid}. Suppose the path $p$ joins the path $l \in L$ and $s \in S$. Let $l = l'xl''$ and $s = s'ys''$ where $x$ and $y$ are the endpoints of $p$, and $s'$ is equal or longer than $s''$. Then $l$ is replaced by $l'ps'$, and $s$ by $s''$, while we discard the path $l''$.   

Let $\hat{L}$ and $\hat{S}$ be the long and short paths that are joined, and $L^*$ be the part of long paths that get discarded. Besides $P$ being \emph{valid}, Aggarwal and Anderson~\cite{Aggarwal87} want two more properties:
\begin{enumerate}
    \item The set of paths $P$ are maximal.  
    \item There is no path between $L^*$ and $S-\hat{S}$.
\end{enumerate}

Suppose that at least $\frac{1}{12}k$ of the paths are joined, and $Q$ remains a separator after the replacement. Then the length of at least $\frac{1}{9}$ of the short paths are reduced by $\frac{1}{2}$. In at most $O(\log n)$ such steps, at least $\frac{1}{4}k$ of the short paths are removed. When this happens, one can stop because the number of paths in the separator has been reduced by a constant fraction.  

If either of the above two conditions fails, then they~\cite{Aggarwal87} can directly find a different $Q'$ such that it consists of less than a constant fraction of paths than $Q$. 

\subsubsection{Path Reduction Algorithm}
\label{pathreductionalgorithm}
The bottleneck for the work in \cite{Aggarwal87} is the algorithm that finds a maximal $P$. They reduce this problem to the problem of minimum weight perfect matching, and solve it in $\poly(\log n)$ depth and using $\Omega(n^3)$ work via known perfect matching algorithms~\cite{karp1985constructing}. In our paper, we crucially want work-efficient algorithms, which use $\tilde{O}(m)$ work: 
\begin{lemma}
\label{Lemma reduce path}
    Given a separator $Q$ that consists of $k$ paths, with $k > 48\sqrt{n}$, there is an algorithm that finds $Q'$ that consists of at most $\frac{47}{48}k$ paths in $O(\sqrt{n} \log^7 n)$ depth and $O(m \log^6 n)$ work .
\end{lemma}

To prove this lemma, we find a set of vertex disjoint paths $P$ that are valid (as defined above) with the exception that some of the paths in $P$ have no endpoint in $S$. We divide $P$ into $P_1$ and $P_2$ where $P_1$ are the ``matched'' paths that have the other end on $S$, while $P_2$ are the ``unmatched'' paths, meaning that their other ends are not on $S$.

Similarly, let $\hat{L}_1$ be the paths of $L$ that are joined with $S$ using $P_1$. Suppose that the long path $l = l'xl'' \in  \hat{L}_1$ and the short path $s = s'ys'' \in S$ are joined by $p\in P_1$. We would like to learn whether $|s'| \geq |s''|$ so that we can determine which part of the short path to join the long path. Since $s$ is provided as a doubly-linked list, we make a copy of $s$ and only keep one direction of the doubly-linked direction. Without loss of generality, let $s= s'ys''$ be the direction that is kept. Then we assign the value $1$ to each element on the copied linked list, and we invoke \cref{lemma:andersonmiller} on the copied $s$ to learn $y$'s rank on the list. If $y$'s rank is greater or equal to $\frac{1}{2}$ of the rank of the last vertex on the list, then  $|s'| \geq |s''|$. Otherwise, $|s''| > |s'|$. This operation can be performed simultaneously for all paths with work proportional to the length of the paths and depth $O(\log n)$. Without loss of generality we assume $|s'| \geq |s''|$, we then replace $l$ by $l'ps'$, $s$ by $s''$, and we add the discarded part $l''$ to $L^*_1$.

For the path $l = l'xl''\in \hat{L}_2$ and the unmatched path $p \in P_2$. We replace $l$ with $l'p$, discard $l''$. We add the discarded part $l''$ to $L^*_2$.  Denote the discarded parts as $L^* = L_1^* \cup L_2^*$, the long paths that are attached to $P$ as $\hat{L} = \hat{L}_1 \cup \hat{L}_2$, and the short paths that are attached to $P$ as $\hat{S}$. Let $D = V-P-L-S$ be the set of vertices not on any paths. A picture for illustration is shown in \cref{fig:LnS}. The three properties we need are 
\begin{enumerate}
    \item The set of paths $P$ are maximal in the sense that there is no path from $L - \hat{L}$ to $S-\hat{S}$ with all internal vertices being in $D$.   
    \item There is no path between the discarded part $L^*$ and $S-\hat{S}$ with all internal vertices being in $D$.
    \item The number of unmatched paths is small; concretely the number of paths in $P_2$ is less or equal to $\frac{1}{48}k$.
\end{enumerate}
\begin{figure*}[th]
    \centering
    \includegraphics[scale=0.4]{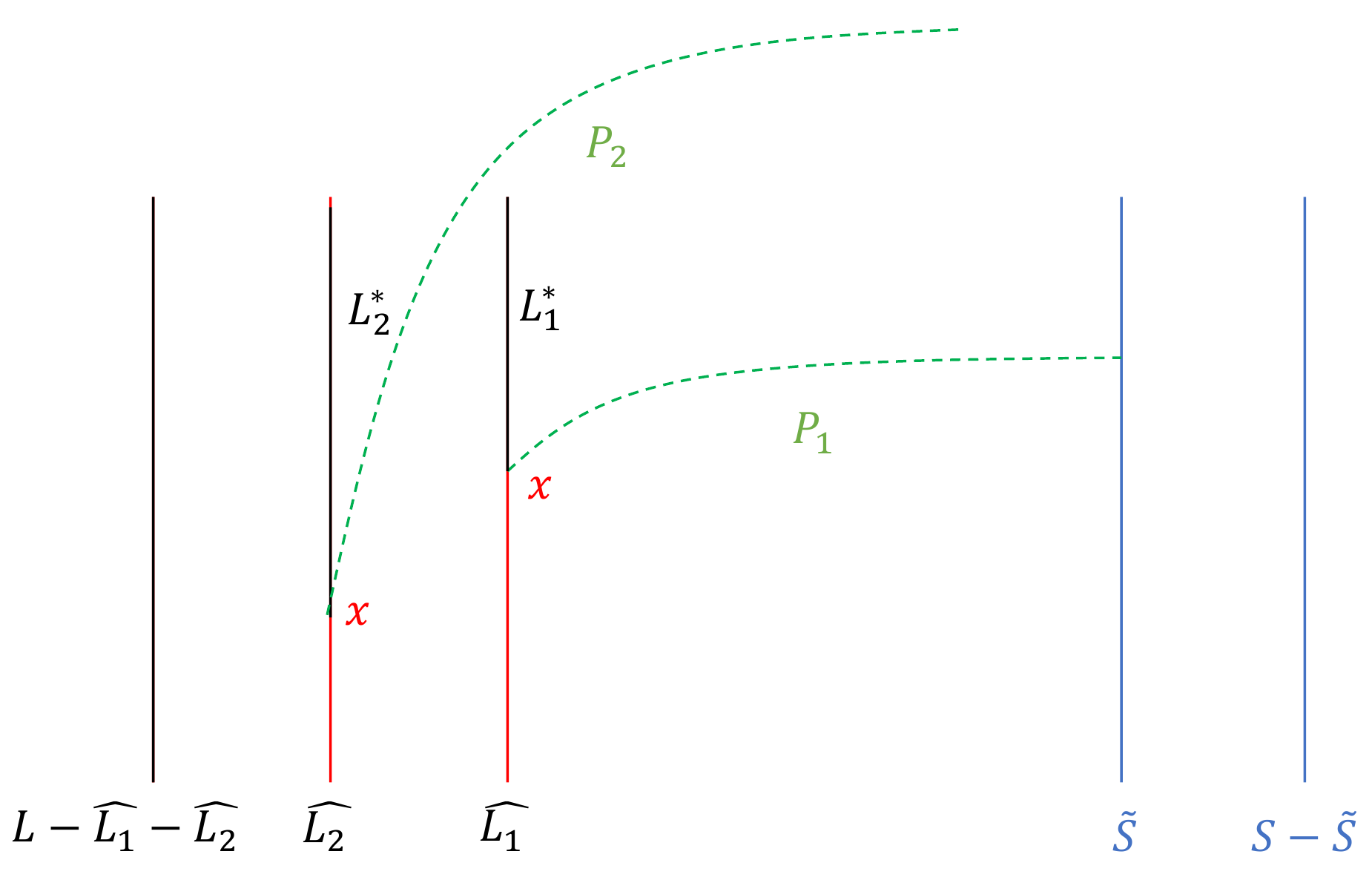}
    \caption{Merging long and short paths.}
    \label{fig:LnS}
\end{figure*}

Suppose that these three conditions hold; we argue that either the new $Q = L\cup P\cup S - L^*$ remains a separator or we can directly find a different $Q'$ consisting of less than a constant fraction of paths in $Q$. There are potentially two problems that can arise. In each case, we find a different separator with a constant reduction in the number of paths. We list the two potential problems below:

\begin{enumerate}
    \item The discarded parts of $L^*$ caused $Q$ to no longer be a separator.
    \item The number of matched paths is too small, meaning the number of paths in $P_1$ is less than $\frac{1}{12}k$.
\end{enumerate}

Aggarwal and Anderson~\cite{Aggarwal87} show that in either case, we can immediately find another separator consisting of less than $\frac{47}{48}k$ paths. In our algorithm, we can follow the same solutions for these issues. \fullOnly{To have a self-contained description, we reproduce their explanations in \Cref{app:singularities}.}\shortOnly{For a review of their explanations, please see the full version of our paper.} 
%
%
For the rest of this section, we can assume that neither problem arises. Thus $Q$ remains a separator and at least $\frac{1}{12}k$ short paths belong to $\hat{S}$. Recall that originally, there are $\frac{3}{4}k$ short paths. So at least ${(\frac{k}{12})}/{(\frac{3k}{4})} = \frac{1}{9}$ fraction of the original short paths have their size cut in half. Thus, in $O(\log n)$ iterations, all short paths will be absorbed. 

We will show in \cref{lemmapathmerging} that finding the desired set $P$ can be done in $O(\sqrt{n}  \log^6 n)$ depth and $O(m \log^5 n)$ work. For short paths of the form $s=s'ys''$, determining whether $|s'| \geq |s''|$ can be done in $O(\log n)$ depth and $O(n \log n)$ work across all short paths. Since we repeat the procedure at most $9\log n$ times, we prove the work and depth bound of \cref{thm:separator}. 
 
\subsection{Path Merging}
\label{subsec:path_merging}
We first give a high-level description of our path-merging algorithm. In the beginning, each long path chooses one end as its head. Then each path simultaneously tries to extend itself from the head until it reaches a short path. If a path cannot extend, it kills the head vertex and backtracks one vertex. We continue the procedure until the number of "active" long paths is less than $\sqrt{n}$. This way, we ensure that roughly at least $\sqrt{n}$ of the work is parallelized, which in turn bounds the total depth of the algorithm.

\begin{lemma}
\label{lemmapathmerging}
    Assume $k \geq 48\sqrt{n}$. There is a parallel algorithm that finds the set of paths $P$ satisfying the following three properties using depth $O(\sqrt{n} \log^6 n)$ and work $O(m \log^5 n)$.
\begin{enumerate}
    \item The set $P$ of paths is maximal, i.e., there is no path from $L - \hat{L}$ to $S-\hat{S}$ such that the internal vertices of the path are in $D$.   
    \item There is no path between the discarded part $L^*$ and $S-\hat{S}$ such that the internal vertices of the path are in $D$.
    \item The number of unmatched paths is small; concretely the number of paths in $P_2$ is less or equal to $\frac{1}{48}k$.
\end{enumerate}
\end{lemma}
 We work on an auxiliary graph $G'$ where we contract each short path $s\in S$ into a single vertex $v_s$. Throughout the process, each vertex in $G'$ is in one of three states, namely \emph{available}, \emph{contained in a (long) path}, or \emph{dead}. Initially, all vertices not contained in a long path are \emph{available}.
 For each long path $l$, we pick an arbitrary end as the head of that path, denoted as $u_l$. Now, in each step, a head vertex $u_l$ either gets matched to a neighbor $v_l$ that is still \emph{available} or it doesn't get matched. If $u_l$ gets matched to $v_l$, then $v_l$ joins the path $l$, and $v_l$ becomes the new "head" of the path $l$.
 If $u_l$ does not get matched, then only because of two possible reasons: Either, the node $u_l$ corresponds to a contracted short vertex $v_s$. If that happens, we say that $l$ succeeded. If $l$ succeeded, then its head vertex does not get matched and $l$ remains unchanged.
 Otherwise, if $u_l$ does not correspond to a contracted short vertex $v_s$, then $u_l$ has not been matched because all available neighboring vertices of $u_l$ have been matched to different head vertices. In that case, $u_l$ dies and is removed from the path $l$, and its predecessor $w_l$ in the path $l$ becomes the new head $l$. If $u_l$ was the only node in $l$, then the path $l$  does not participate in the matching process anymore. 
 If the number of long paths attempting matching (or equivalently, the number of head vertices not corresponding to a contracted vertex $v_s$) is less than $\sqrt{n}$, then the process terminates, and we go to the postprocessing phase described below.  

For a path $l$ that successfully joins a short path, let $l_{old}$ be its form before the above update, and $l'_{new} = \{v_1 \codots v_k\}$ be its form after the update. The last vertex $v_k \in G'$ on $l_{new}$ corresponds to a contracted short path $s\in S$ in the original graph $G$. Thus, the edge $(v_{k-1},v_{k}) \in G'$ corresponds to an edge $(v_{k-1},v'_k) \in G$. We replace $v_k$ with $v'_k$ to get the corresponding path in $l_{new}$ in the graph $G$. Then $p_1 = (l_{new}- l_{old})\cup x$, where $x$ is defined to be the only vertex in $l_{old}\cap l_{new}$ that was a head vertex during the path merging algorithm. So $p_1$ belongs to $P_1$ and path $l_{old}$ belongs to $\hat{L}_1$. 

Similarly, for a path $l$ that still participates in the matching process when the algorithm terminates, let $l_{old}$ be its old form and $l_{new}$ be its form after the update. Let $p_2 = (l_{new} - l_{old})\cup x \in P_2$ where $x$ is defined analogously. Path $l_{old}$ belongs to $\hat{L}_2$.
For a path $l$ that is not participating because all the vertices on the paths are dead, it belongs to $L - \hat{L}_1-\hat{L}_2$. 

We see that $P_1$ has one end on a long path and one end on a short path. And $P_2$ has one end on a long path and another on a vertex not in $L$ or $S$. Now we show why the paths we find satisfy the above three properties. Property 3 holds because $k \geq 48\sqrt{n}$ and the number of paths in $P_2$ is less than $\sqrt{n}$. So we have that the number of paths in $P_2$ is less than $\frac{1}{48}k$. For properties 1 and 2, we observe that all vertices in $L-\hat{L}_1-\hat{L}_2$ or $L^*$ are dead. Notice that a vertex can only become dead if it was a head vertex of a long path but failed to get matched in the matching process. Hence, we get properties 1 and 2 from the following lemma.
\begin{lemma}
Suppose that a vertex $v$ becomes dead during the above algorithm, then there is no path from $v$ to a vertex in $S - \hat{S}$ such that the internal vertices of the path are in $D$.  
\end{lemma}

\begin{proof}
      For a vertex $v$ to die, it must be a head vertex of a long path during the matching process. It becomes dead if all its neighbors are dead or they joined some other long paths. For the sake of contradiction, suppose that there is a path $p = (v_1 = v, v_2 \codots v_k=v_s)$ from $v$ to a contracted vertex $v_s$ formed from a short path with all internal vertices being in $D$. As $v_1$ is dead and $v_k$ is available, there exists some index $i$ with $v_i$ being dead and $v_{i+1}$ being available. Since $v_{i+1}$ is available at the end of the matching process, it must have been in the available state throughout the whole matching process. This yields a contradiction as when $v_i$ was attempting matching, it could have been matched to $v_{i+1}$ which means $v_i$ would not be dead. 
\end{proof}

\begin{proof}[Proof of \cref{lemmapathmerging}]
We have shown that the output satisfies all the guarantees of \cref{lemmapathmerging}. \cref{lem:pram_implementation}, proven in the next section, shows that there exists a parallel implementation of the procedure with work $O(m \log^5 n)$ and depth $O(\sqrt{n} \log^6 n)$, which finishes the proof of \cref{lemmapathmerging}.
\end{proof}

\subsection{Parallel Implementation}
\label{subsec:pram_implementation}
It remains to discuss the parallel implementation of the procedure described in \cref{subsec:path_merging}. In particular, the remaining part of this section is dedicated to proving the lemma below. We note that some parts of the matching procedure are nontrivial. The main reason the parallel implementation is nontrivial is that the number of steps for computing the paths can be up to $\Theta(\sqrt{n})$. However, we can only afford $O(m\log^5 n)$ work overall. Thus, the algorithm cannot afford to read the whole input in each iteration, and we should ensure that each edge is read only $\poly(\log n)$ times, in an amortized sense.

\begin{lemma}
    \label{lem:pram_implementation}
    There is a parallel algorithm that implements the procedure described in \cref{subsec:path_merging} with $O(m\log^5 n)$ work and $O(\sqrt{n}\log^6 n)$ depth.
\end{lemma}

Recall that, at any point in time, a node is in one of three states and can change its state at most twice.
Moreover, in each step, the number of vertices changing their state is equal to the number of head vertices attempting to get matched.
As the process stops once less than $\sqrt{n}$ head vertices attempt to get matched, this implies that at least $\sqrt{n}$ vertices change their state in each step.
Hence, the total number of steps is upper bounded by $2n/\sqrt{n} = O(\sqrt{n})$.
Therefore, it suffices to show that each step can be implemented with $O(\log^6 n)$ depth.
We say that an edge changes its state if one of its endpoints changes its state. In particular, each edge changes its state $O(1)$ times. Therefore, it suffices to show that each step can be implemented with work $O(N_{change}\log^5 n)$, where $N_{change}$ is the total number of vertices and edges changing their state.
Achieving this bound is nontrivial; just reading all available neighbors of a vertex $u_l$ trying to get matched might already exceed it.
Thus, our matching routine makes use of a data structure that allows to efficiently get access to a subset of $u_l$'s neighbors that are still available. We use the data structure from the lemma below.

\begin{restatable}{lemma}{decrementalgraph}
\label{lem:decremental_graph}
There is a data structure with the following guarantees:
The initial input is a graph $G = (\{v_1,v_2,\ldots,v_n\},\{e_1,e_2,\ldots,e_m\})$ with $n \geq 2$ where all vertices are active in the beginning. The data structure supports the following operations after initialization:
\begin{itemize}
        \item \textbf{MakeInactive$(\{i_1, i_2, \ldots, i_k\})$} takes an array consisting of $k \geq 1$ distinct indices between $1$ and $n$ and marks the corresponding vertices as inactive. This operation can be done in $O((k + \sum_{j=1}^k deg_G(v_{i_j})) \log n)$ work and $O(\log n)$ depth.
        \item \textbf{Query($i_1,i_2,\ldots,i_k,t$)} takes an array consisting of $k$ distinct indices between $1$ and $n$ and a number $t$. The output is an array $a_j$ for every $j \in [k]$ containing $t$ distinct \emph{active} neighbors of vertex $v_{i_j}$. If $v_{i_j}$ has fewer than $t$ active neighbors, then $a_j$ contains one entry for each active neighbor of $v_{i_j}$. The work is $O(k \cdot t \cdot \log n)$, and the depth is $O(\log n)$.
    \end{itemize} 
Initialization takes $O((m+n)\log n)$ work and $O(\log n)$ depth.
\end{restatable}
\fullOnly{The data structure uses a common technique, and we defer the detailed description of it to \cref{sec:appendix_missing_proofs}.}\shortOnly{The data structure uses a common technique, and we defer the detailed description of it to the full version of this paper.} On a high level, the adjacency list of each node $v$ is augmented with a balanced binary tree. The leaves correspond to $v$'s neighbors, and each internal node keeps track of how many neighbors in the corresponding subtree are still active.
In our concrete case, the input graph is the graph $G'$. We also maintain the invariant that a node is available if and only if it is marked as being active in the data structure. 
Note that we can initialize the data structure right at the beginning with $O(m \log n)$ work and $O(\log n)$ depth.
Now, let's consider an arbitrary step of the algorithm. We denote by $U$ the set consisting of all head vertices $u_l$ trying to get matched. 
Our matching procedure builds the matching gradually in $O(\log n)$ phases. In each phase, it uses Luby's deterministic parallel maximal matching algorithm (\cref{lem:lubylemma}) as a black box on a graph that is constructed with the help of the data structure. In more detail, in phase $i$, each node $u \in U$ that has not been matched in previous phases first selects $2^i$ arbitrary neighbors that are still available (in particular, which haven't been matched in previous phases). If $u$ has fewer such neighbors, it selects all of them. By making use of the data structure, we can do the selection using $O(2^i\log n)$ work per node $u \in U$ that has not been matched and $O(\log n)$ depth. Let $H_i$ be the bipartite graph where one side of the bipartition consists of all nodes in $U$ that have not been previously matched, and the other side consists of all available nodes that have been selected by at least one node. Moreover, there is an edge between a head vertex $u$ and an available vertex $v$ if and only if $u$ has selected $v$.
We then compute a maximal matching of $H_i$ in $O(|V(H_i)|+|E(H_i)|\log^5 n)$ work and $O(\log^5 n)$ depth using the algorithm of \cref{lem:lubylemma}. In particular, one has $O(2^i\log^5 n)$ work per vertex in $U$ that has not been matched before. Then, the data structure is updated by marking all previously available vertices that have been matched as inactive. Note that we can pay for this operation by charging $O(\log n)$ to each such vertex and $O(\log n)$ to each incident edge. We can do this charging as all previously available vertices that have been matched change their state (and thus also the incident edges).
Also, after each phase, we remove vertices from $U$ that don't have any available neighbors.
Both the work and the depth of the algorithm are dominated by the $O(\log n)$ invocations of Luby's deterministic maximal matching algorithm. In particular, the overall depth is $O(\log n) \cdot O(\log^5 n) = O(\log^6 n)$. To upper bound the work, consider some vertex $u \in U$. First, consider that $u$ has been matched in some phase $i$. Then, informally speaking, the algorithm has done $O(2^i \log^5 n)$ work on behalf of $u$ in all phases combined. If $u$ has been matched in the first phase, then we charge $O(\log^5 n)$ work to the node that $u$ matched with. If $u$ has not been matched in the first phase, then at least $2^{i-1}$ neighbors of $u$ got matched in phase $i-1$. Thus, at least $2^{i-1}$ edges incident to $u$ change their state, and therefore $u$ can charge each such edge $O(\log^5 n)$.
We can use a similar charging argument if $u$ has not been matched. 
Therefore, the work in each step is $O(N_{change}\log^5 n)$.
This shows the work and depth bound of \cref{lem:pram_implementation}.
It remains to argue correctness. In the last phase $i$, we have $2^i \geq n$. Therefore, each vertex $u \in U$ that has not been matched before selects all its neighbors that are still available. Thus, if the maximal matching in $H_i$ does not match $u$, then all its available neighbors have been matched to other head vertices, which shows that the final matching indeed satisfies the guarantees stated in the previous section.

\section{Constructing an initial segment from the separator paths}
\label{sec:initialsegment}
In this section, we prove \Cref{thm:merging}, which shows how we can add the $O(\sqrt{n})$ paths in the separator to the partial DFS one by one, using $\poly(\log n)$ depth for each path, and $\tilde{O}(m)$ work overall. For the sake of readability, we first restate the lemma.

\merging*

We want to construct an initial segment $T'$ that contains all the vertices in $Q$---the set of all vertices of the separator paths. We do this by absorbing paths of $Q$ into $T'$ one by one. To have a work-efficient algorithm, we want to ensure that in each iteration of absorbing a path, the total work is near-linear with respect to the number of edges adjacent to the path that got absorbed into $T'$. 

Our algorithm uses a batch-dynamic parallel data structure. We next describe the interface of this data structure and use that to provide a proof for \Cref{thm:merging}. The actual data structure that proves this lemma is presented in \Cref{sec:dataStructure}.

\begin{restatable}{lemma}{datastructure}
\label{lemma:parallel_data_structure}   
Given a graph $G=(V,E)$, a separator $Q$ that consists of some paths, and a root $r\in G$ that forms the initial partial tree $T'$, there is a data structure that supports the following operations.
\begin{itemize}
    \item \textbf{FindCC()}  returns a connected $C \subseteq G-T'$ such that $C\cap Q \neq \emptyset$, if no such $C$ exists, the function returns $Success$.
    \item \textbf{LowestNode($C$)} takes a connected component $C \subseteq G-T'$, and returns the vertex $v\in C$ that is adjacent to a vertex in $x\in T'$. The vertex $x$ is the unique vertex with the lowest depth that is adjacent to $C$.
    \item \textbf{FindPathS2P($C$, $x$)} takes a connected component $C\subseteq G-T'$ and a vertex $x\in T$ and returns a path $p$ from $x$ to a vertex $q \in Q$. All the vertices in $p$ are not in $Q$ except for $q$. 
    \item \textbf{BatchDelete($p$}) takes a path $p$ consists of vertices $(v_1\codots v_k)$ and deletes the vertices from $G-T'$. 
\end{itemize}
Moreover, here are the work and depth of the above operations.
\begin{itemize}
    \item \textbf{FindCC()} has work and depth $O(1)$. 
    \item \textbf{LowestNode($C$)} has work and depth $O(1)$.
    \item \textbf{FindPathS2P($C$, $x$)} has depth $O(\log n)$ w.h.p. If the function returns a path $p$, the work is $O(|p|\log n)$ where $|p|$ is the number of vertices on the path $p$
    \item \textbf{BatchDelete($p$)} has work $O(|E(p)| \log^3 n)$ in expectation and depth $O(\log^2 n)$ w.h.p. where $|E(p)|$ are the number of edges adjacent to the vertices in $p$.
\end{itemize}
\end{restatable}

Having this data structure, we can now prove \Cref{thm:merging}.

\begin{proof}[Proof of \cref{thm:merging}]

To absorb the separator into the tree, we sequentially find a path $p$ from a vertex $q\in Q\cap (G-T')$ to a vertex $x \in T'$, such that all the internal vertices of the path are in $G-T'$. Suppose that $q$ belongs to the path $l = l'ql''$ in the separator, and $|l'| \geq |l''|$. Then we take the longer half $l'$ and incorporate it into the tree $T'$, by adding the path $pql'$ to the tree $T'$. For the new $T'$ to remain an initial segment, we need $x$ to be the lowest vertex in $T'$ adjacent to the connected component containing $q$. 

For each path $l = l'ql''$, we need to learn whether $|l'| \geq |l''|$, so that we can determine which part of the path to absorb in the initial segment. Since $l$ is provided as a doubly-linked list, we make a copy of $l$ and only keep one direction of the doubly-linked direction. Without loss of generality, let $l= l'yl''$ be the direction that is kept. Then we assign the value $1$ to each element on the copied linked list, and we invoke \cref{lemma:andersonmiller} on the copied $l$ to learn the rank of $q$ in the list. If this rank is greater or equal to $\frac{1}{2}$ of the rank of the last vertex on the list, then  $|l'| \geq |l''|$. Otherwise, $|l''| > |l'|$. This check can be performed using depth $O(\log n)$ and work proportional to the length of the path.

We also need to absorb the path $pql'$ to $T'$, and vertices on $pql'$ should learn their depth in $T'$. Suppose the absorption is through the edge $(x,y)$ with $y\in T'$ where $x$ is the first vertex on the path $pql'$. We next make each vertex in $pql'$ learn its depth in $T'$, using $O(\log n)$ depth and work proportional to the length of $pql'$. For a prefix sum computation, we initiate vertex $x$ with the value $depth_y$---i.e., the depth of vertex $y$ in the existing partial tree $T'$---and we initiate each remaining vertex on the $pql'$ path with value $1$. Then we invoke \cref{lemma:andersonmiller} on the path to compute the prefix sum on the list. As a result, all vertices on $pql'$ learn their depth in $T'$. 
This operation can be done in $O(\log n)$ depth and $O(|pql'|)$ work where $|pql'|$ is the length of the path. 

We repeatedly call FindCC() to find a connected component containing a vertex from $Q\cap (G-T')$, then find $x$ using LowestNode($C$) and call FindPathS2P($C$, $x$) to find the desired path $p$. Then we determine whether $|l'| \geq |l''|$, and join the path $pql'$ to $T'$. Finally, we call BatchDelete(\{$pql'$\}) to delete the vertices from $G-T'$. 

\cref{lemma:parallel_data_structure} shows that the first two operations have depth and work $O(1)$, the third operation has work $O(|p| \log n)$ and depth $O(\log n)$, and the last operation has $O(|E(p)| \log^3 n)$ amortized work in expectation and depth $O(\log^2 n)$, with high probability. Here $|E(p)|$ are the number of edges adjacent to the vertices in $p$. Determining whether $|l'| \geq |l''|$ and joining the path $pql'$ to $T'$ can be done in $O(\log n)$ depth for one joining and $O(n)$ work in total overall joining operations. Since we repeat the above sequence of operations $O(\sqrt{n}\log n)$ times, the algorithm uses $O(\sqrt{n} \log^3 n)$ depth w.h.p. and $O(m \log^3 n)$ work.
\end{proof}

\section{Data Structures}
\label{sec:dataStructure}
The data structure we use is based on a combination of the ones developed by Acar et al. in \cite{acar2020parallel} and \cite{acar2019parallel}, though we also need some further modifications. 


\label{batchdynamicdatastructure}
We work on $G-T'$, and this graph undergoes vertex and edge deletions. During the construction of the initial segment in \cref{sec:initialsegment}, we need to repeatedly find a path from the lowest vertex in the partial tree $T'$ to a vertex in the separators, such that the path is made of internal vertices in $G-T'$. To do that efficiently, we maintain the connectivity structure of $G-T'$. We would like to keep one spanning tree for each connected component of $G-T'$. The are two main problems that we want to solve using this data structure. The first is that after joining a path to the partial DFS tree $T'$, thus deleting its vertices from $G-T'$, we need to update the connectivity structure of $G-T'$ and in particular, as some edges get deleted from the respective tree, we might have to find replacement edges for them. The second problem is to maintain the connectivity structure of $G-T'$ such that we can answer path queries of the following type efficiently: Given a set of vertices $C$ and a vertex $x$, we need to report a path in $G-T'$ connecting $x$ to a vertex in $C$, using $\poly(\log n)$ depth and work proportional to the path length.

In the batch-dynamic setting, a batch of updates (or queries) are applied simultaneously; in our case, we will have batches of deletions to $G-T'$. For each such batch, we would like the work to be near-linearly proportional to the number of updates while keeping the depth $O(\poly\log n)$. For this purpose, we use a modified and combined version of the data structures provided in \cite{acar2019parallel} and \cite{acar2020parallel}. We first provide a brief recap of these and then present the combined and adapted data structure that we need. 

\subsection{Connectivity and rake-and-compress data structures}
 Before recapping the data structures of \cite{acar2020parallel, acar2019parallel}, let us remark on a small subtlety: their algorithms are written as supporting edge deletions. However, we can generally treat vertex deletions as deleting all edges adjacent to the vertex. 

\subsubsection{Parallelized connectivity data structure Algorithm}
\label{parallalconnDS}
Consider a graph undergoing (batches of) edge deletions, and suppose we want to maintain a spanning tree for each connected component of it. Acar et al.~\cite{acar2019parallel} provide a solution for this, which is essentially a parallelized version of the sequential algorithm developed by Holm, de Lichtenberg, and Thorup (HDT)~\cite{holm2001poly}. The HDT algorithm maintains a maximal spanning forest, certifying the graph's connectivity. We also note that although the work in the following lemma is stated in expectation, by increasing the work by a factor of $O(\log n)$, the work bound also holds w.h.p. 
\begin{lemma}
\label{oldhdtlemma}
    There is a parallel batch-dynamic connectivity data structure that maintains a maximal forest for a graph undergoing vertex and edge deletions. Given any batch of edge deletions, the data structure uses $O(\log^2 n)$ expected amortized work per edge deletion. The depth to process a batch of edge deletions is $O(\log^3 n)$ w.h.p.
\end{lemma}
The main challenge is that, when an edge of the spanning forest is deleted, which essentially breaks the tree of the component into pieces, we need to see if there is a \textit{replacement} edge that connects the pieces. Furthermore, we might have to do several of these simultaneously, for a batch of edge deletions. To solve this efficiently, the HDT algorithm maintains a set of $\log n$ nested forests. The topmost level of the nest forest represents a spanning forest of the entire graph. Each level contains all tree edges stored in levels below it. A key invariant we keep is that the largest component size at the level $i$ forest is at most $2^i$. When an edge is searched and fails to become a replacement edge in the forest, we decrease the level of that edge by 1. This way, we ensure that an edge is searched at most $\log n$ times before it is deleted from our graph. Acar et al.~\cite{acar2019parallel} parallelize the task of finding replacement non-tree edges by examining multiple potential replacements at once, which gives the lemma we stated above.

\subsubsection{Parallalized Rake and Compress Tree}
\label{parallalRC}

The rake and compress operation for a static tree is a simple recursive procedure that allows one to ``process" a tree in $O(\log n)$ simple iterations. The rake operation removes all leaves from the tree, except in the case of a pair of adjacent degree-one vertices where it removes the one with a smaller vertex id. The compress operation removes an independent set of vertices of degree two that are not adjacent to leaves. It is desired that this independent set has size within a constant fraction of the maximum independent set (in expectation). A simple analysis shows that after $O(\log n)$ iterations of rake and compress, the tree shrinks to a single vertex. This is because a constant fraction of the vertices in a forest are either leaves or degree 2 vertices due to the degree constraint.

\paragraph{Rake-and-compress as a low-depth hierarchical clustering} The rake and compress can be viewed as a recursive clustering process. A cluster is a connected subset of vertices and edges of the original forest. We note that a cluster may contain an edge without containing both of its endpoints. The boundary vertices of a cluster $C$ are the vertices $v\in C$ that are adjacent to an edge $e \in C$. The vertices and edges of the original forest form the \textit{base clusters}. Initially, each vertex and each edge form its own cluster. In the course of rake and compress, clusters are merged using the following rule: Whenever a vertex $v$ is removed during rake or compress operations, all of the clusters with $v$ as a boundary vertex are merged with the base cluster containing $v$; we say $v$ represents the new cluster. The children of a vertex are clusters that merged together to form it, which we store in an adjacency list. Thus we will have a collection of hierarchical forest $\{T_1 \codots T_k\}$ where $T_1$ is the original forest, and $T_{i+1}$ is generated from $T_i$ after one iteration of the rake-and-compress. $T_k$ is the final forest in which every connected component in the original forest is clustered into a single cluster. An example rake and compress tree are depicted in \cref{fig:Rake and Compress Tree}
\begin{figure*}[th]
    \centering
    \includegraphics[scale = 0.38]{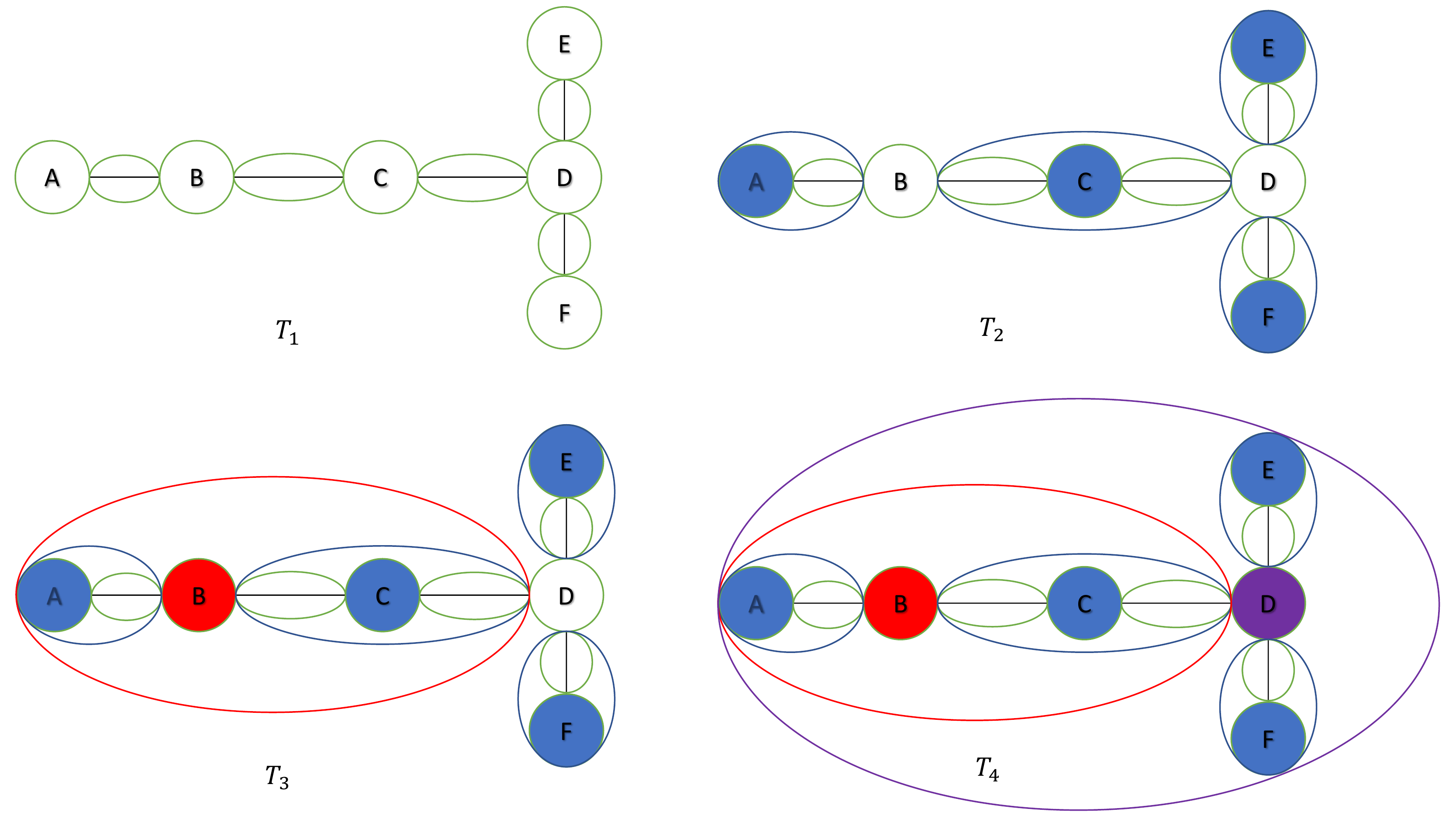}
    \caption{\footnotesize A recursive clustering of a tree. The tree consists of vertices $\{A,B,C,D,E,F\}$ and those vertices are connected by black edges. We use circles to represent different clusters. 
    $T_1$: The vertices and edges in the original tree form the green clusters. 
    $T_2$: The blue vertices are removed to form the blue clusters. Vertices $\{A,E,F\}$ are removed through the rake operation. Vertex $\{C\}$ is removed through the compress operation. 
    $T_3$: $\{B,D\}$ are both leaf vertices in the tree. Through tie-breaking, $\{B\}$ is removed to form the red cluster. 
    $T_4$: The only (leaf) vertex $D$ in the tree is removed through the rake operation to form the purple cluster.
    }
    \label{fig:Rake and Compress Tree}
\end{figure*}

\smallskip
We need a dynamic version of this, which will allow us to answer some type of queries in the tree, e.g., reporting a path in the tree between two nodes. Consider a forest that undergoes edge deletions and insertions (with the promise that all edges present after any of these updates form a forest). We would like to maintain the results of the rake and compresses operations during the $O(\log n)$ iterations, while the forest undergoes these updates.

Acar et al.~\cite{acar2020parallel} first observe that the decision of whether to remove a vertex $v$ only depends on $v$'s neighbors $N(v)$ and the leaf status of $N(v)$. Then they argue that after one edge insertion/deletion, only a constant number of vertices will have neighbors that change their leaf status. These vertices are called affected vertices. Meanwhile, for all other vertices, the decisions of whether to remove themselves are unaffected. In each subsequent iteration of the rake-and-compress process, the number of affected vertices due to that edge insertion/deletion grows by a constant additive factor. Thus the work induced by one edge insertion/deletion is only $O(\log n)$, yielding the following lemma. Like before, we note that although the work in the following lemma is stated in expectation, by increasing the work by a factor of $O(\log n)$, the work bound also holds w.h.p. 

\begin{lemma}
\label{oldrclemma} There is a dynamic data structure that maintains the hierarchical forests $\{T_1 \codots T_k\}$ where $T_1$ is the original forest, and $T_{i+1}$ is generated from $T_i$ after one iteration of the rake-and-compress process. And $T_k$ is the final forest in which every connected component in the original forest is clustered into a single cluster. For a $n$-node forest that undergoes edge insertions and deletions, processes batch insertions and deletions of $k$ edges in $O(k \log n)$ work in expectation and $O(\log(n)\log^*(n))$ span w.h.p.  
\end{lemma}

\subsection{Overview of the combined data structure} 
The combined and adapted data structure we present proves \Cref{lemma:parallel_data_structure}. The rest of this section is dedicated to proving \Cref{lemma:parallel_data_structure}. We first provide an overview of the data structure, and we then discuss each of the operations and their complexities.

To explain the data structure, let us briefly summarize what we will use it for. Our data structure is meant to be working on $G-T'$, the graph induced by vertices that are not in the current initial segment. We construct and maintain a parallelized HDT connectivity forest (as overviewed in \Cref{parallalconnDS}). In addition, we keep a rake-and-compress representation (as overviewed in \Cref{parallalRC}) for a copy of the HDT forest. The parallelized HDT forest is responsible for keeping a maximal forest of $G-T'$ as it undergoes vertex and edge detection. Whenever some edge deletions happen, the replacement edges found by this parallelized HDT are also fed into the rake-and-compress representations of the trees. The rake-and-compress representation allows one to efficiently answer queries about the component, e.g., finding a path between two vertices. 

We augment the RC-tree with two different flags (from now on, we use the phrase RC-Tree to refer to the tree's rake-and-compress representation). The first flag is the separator flag which helps us to determine whether a connected component contains a vertex from the separator $Q$. In the RC-Tree, each vertex has a flag indicating whether it is in the separator $Q$. In the original tree $T_1$, a vertex has a \textit{separator flag} if it is in $Q$. When removing a vertex $v$ from the RC tree, we merge all the clusters with $v$ as the boundary vertex to form a new cluster. If any of the above clusters or $v$ has a separator flag, the newly formed cluster will have the separator flag. Moreover, in the adjacency list, we sort the children cluster such that the children with separator flags appear in front of the children without separator flags. Finally, the clusters representing connected components are sorted into a linked list where the connected components with separator flags appear before the ones without the flags.

The second augmentation is the lowest neighbor augmentation which helps us find a connection vertex of a connected component $C$ to the partial tree $T'$. Before the first stage of the rake-and-compress process, if a vertex $v$ has neighbors in the partial tree $T'$, it is augmented with $(v,d_v)$ where $d_v$ is the depth of the lowest depth tree neighbor of $v$. If $v$ has no tree neighbor, it is augmented with \textit{None}. When merging clusters formed by removing $v$, suppose any of the clusters or $v$ is augmented with anything other than \textit{None}. Let $\{(v_1,d_{v_1})\codots (v_k,d_{v_k})\}$ be the augmentations of $v$'s children clusters, the cluster formed by removing $v$ is augmented with $(v_i,d_{v_i})$ where $d_{v_i}$ has (one of) the lowest depth among $d_{v_1} \codots d_{v_k}$. Otherwise, it is augmented with \textit{None}. If $v$ is augmented with $(v_i,d_{v_i})$, then $v_i$ has a neighbor $y\in T'$ such that $y$ has the lowest depth among all tree neighbors of vertices in the cluster formed by $v$. 

Next, we discuss the complexities of the operations provided by the data structure. We first discuss FindCC(), LowestNode(), and BatchDelete(), as they are simpler and shorter. Then, in a separate subsection, we discuss FindPathS2P().

\subsection{FindCC(), LowestNode(), and BatchDelete()}
FindCC() has work and depth $O(1)$. This is because the connected components of $G-T'$ are sorted such that connected components with separator flags appear before those without. So we can just check the separator flag of the first connected component in the list. If it has the flag, then we return the connected component. Otherwise, the function returns $success$.

LowestNode($C$) has work and depth $O(1)$ because it is augmented with $(v,d_v)$.  As described above, $v$ is the unique vertex with the lowest depth that is adjacent to a vertex contained in the cluster, and that cluster is the connected component $C$. 

For BatchDelete($p$), vertices in $G-T'$ adjacent to $p$ learn the depth of their new tree neighbors, so their lowest neighbor augmentations are updated. Suppose vertices on $p$ have $k$ total neighbors. \cref{oldhdtlemma} shows that deleting edges from the parallelized HDT and finding the replacement edges have amortized work $O(k\log^3 n)$. Since we can find a total of at most $k$ replacement edges, by \cref{oldrclemma}, inserting $k$ edges to the RC-Tree can be done in $O(k \log(1 + n/k))$ work in expectation and $O(\log(n)\log^*(n))$ depth w.h.p. This is because the vertices with augmentation updates are the vertices adjacent to the deleted edges.

\subsection{FindPathS2P}
We first describe how to report a path between two given vertices. Then we explain how to report a path between one set of vertices and another vertex, which provides exactly FindPathS2P($C$, $x$). 
\subsubsection{Point-to-point path queries}
We first show that the Point-to-Point queries FindPathP2P($x$,$y$) between two vertices $x$ and $y$ in the RC-Tree can be done using depth $O(\log n)$ and work proportional to $d(x,y)$, where $d(x,y)$ is the distance between $x$ and $y$ in the original tree. 

Recall that the RC algorithm will produce a collection of hierarchical forests $\{T_1 \codots T_k\}$ where $T_1$ is the original forest, and $T_{i+1}$ is generated from $T_i$ after one iteration of the rake-and-compress process. Also, $T_k$ is the final forest in which every connected component in the original tree is clustered into a single cluster. Notice that in forest $T_i$, the edges are either edges in the original forest or clusters formed by the compression of the vertices. The vertices are either vertices in the original forest or clusters formed by the rake-and-compress of the vertices.

\begin{lemma}
\label{lemmast}
    For two vertices $x,y$ in the 1st level RC tree, there is an algorithm that finds the path between $x$ and $y$ in work $O(d(x,y) \log n)$ and depth $O( \log n)$ w.h.p. 
\end{lemma}
We first show that we can answer the path query between adjacent vertices in $T_i$. 
\begin{lemma}
\label{lemmaadjacentst}
    Let $x$, $y$ be two neighboring vertices in the $i^{th}$ level of the RC-Tree $T_i$ connected by the cluster edge $E$. There is an algorithm that finds the path in the original tree between $x$ and $y$ in work $O(d(x,y) \cdot i)$ and depth $O(i)$.
\end{lemma}
\begin{proof}
We prove the above statement by induction. Suppose that $x,y$ are adjacent vertices in $T_1$. Then the path between $x,y$ is just $(x,y)$, and the operation has work and depth $O(1)$. Now in $T_i$, if $E$ is an edge from the original tree, then the path is still $(x,y)$. Suppose that $E$ is a cluster formed by the compression of the vertex $z$ in $T_j$ for $j\leq i-1$. Then $z$ is still neighboring with both $x$ and $y$ in $T_j$. This is because once a vertex $x$ becomes a boundary vertex of a cluster $E$ in $T_j$, it will remain so in subsequent trees until $E$ is merged with either $x$ or another vertex. Moreover, the path between $x$, $y$ in the original tree is the union of the paths formed by $(x,z)$ and $(z,y)$. By induction, we can find the paths between $(x,z)$ and $(z,y)$ in parallel in depth $O(i-1)$, and work $O(d(x,z)\cdot (i-1))$ and $O(d(z,y)\cdot (i-1))$ respectively. Since $d(x,y) = d(x,z)+d(y,z)$, the work of finding the path between $x,y$ in $T_i$ is 
$$O(d(x,z)\cdot (i-1))+O(d(z,y)\cdot (i-1)) + C = O(d(x,y)\cdot i)$$
where $C$ is the constant overhead cost, thus proving \cref{lemmaadjacentst}.
\end{proof}

Consider the general case where $x,y$ are two vertices in the original tree. We recursively generate a list $(X_1 \codots X_k),(Y_1 \codots Y_k)$ where $X_1 = x,Y_1=y$, and 
$$  X_{i+1} = \begin{cases}
        X_i & \text{if $X_i$ is not merged in the formation of $T_{i+1}$}\\
        C & \text{if $X_i$ is merged to form the cluster $C$ in  $T_{i+1}$}
    \end{cases}
$$

We find the largest $i$ such that $X_i \neq Y_i$. Since $X_{i+1} = Y_{i+1}$, we have that $X_i$ and $Y_i$ are children clusters of $X_{i+1}=Y_{i+1}$. Suppose $X_{i+1}$ is formed from the removal of vertex $x_{i+1}$. Finding the path between $x$ and $y$ is equivalent to finding the paths between $x$ and $x_{i+1}$, and between $y$ and $x_{i+1}$. Moreover, $X_i$ and $Y_i$ both have $x_{i+1}$ as one of their boundary vertices. In the following lemma, we show that this can be done in work $O(d(x,x_{i+1})\cdot i)$ and $O(d(y,x_{i+1})\cdot i)$ respectively, thus proving lemma \cref{lemmast}. 

\begin{lemma}
\label{lemmaexposed}
    Consider a vertex $x$ in the RC-Tree with $(x= X_1\codots X_k)$, where 
    $$  X_{i+1} = \begin{cases}
        X_i & \text{if $X_i$ is not merged in the formation of $T_{i+1}$}\\
        C & \text{if $X_i$ is merged to form the cluster $C$ in  $T_{i+1}$}
    \end{cases}
$$
    We denote by $x_i$ the vertex that is removed to form $X_i$. Suppose $X_i$ has $y$ as a one boundary vertex. Then there is an algorithm that finds a path between $x$ and $y$ in work $O(d(x,y)\cdot i)$ and depth $O(i)$.
\end{lemma}

\begin{proof}
The proof is by induction. If $i=1$, $x_1$ cannot have $y$ as a boundary vertex. This is because only clusters formed by the edges can have boundary vertices. If $i=2$, then the path between $x$ and $y$ is just the single edge path $(x,y)$.

Now we consider the case where $i>2$. Suppose $X_i = X_{i-1}$. Then $X_{i-1}$ still has $y$ as its boundary vertex in $T_{i-1}$. This is because once a vertex $v$ becomes a boundary vertex of a cluster $C$ in $T_j$, it will remain so in subsequent trees until $C$ is merged with either $v$ or another vertex. Thus we have the result from induction.

Suppose $X_i \neq X_{i-1}$. We know that $X_{i-1}$ is a child cluster of $X_i$ formed by removing $x_i$ in $T_{i-1}$. In $T_{i-1}$, if $X_{i-1}$ has $y$ as one of its boundary vertices, we have the result from induction. Suppose $X_{i-1}$ has $x_i$ as its boundary vertex and $(y,x_i)$ are neighbors in $T_{i-1}$. In this case, the path between $x$ and $y$ is the union of path between $(x_i, y)$ and $(x,x_i)$. The former path can be found in work $O(d(x_i,y) \cdot (i-1))$ by \cref{lemmaadjacentst}, and the latter path can be found by induction in work $O(d(x,x_i)\cdot (i-1))$, so the total work is still 
$$O(d(x_i,y) \cdot (i-1))+O(d(x,x_i)\cdot (i-1))+C = d(x,y)\cdot i$$
where $C$ is the constant overhead cost, thus proving \cref{lemmaexposed}.
\end{proof}

\subsubsection{Set-to-point path queries}
Recall that in the original tree $T_1$ each vertex is augmented with a flag. In $T_1$, a vertex has the separator flag if it is in $Q$. Moreover, a cluster has the separator if any of its child clusters has the separator flag.

If $x$ has the separator flag, we just return $x$. Otherwise, we recursively generate a list $(x=X_1 \codots X_k)$  where 
    $$  X_{i+1} = \begin{cases}
        X_i & \text{if $X_i$ is not merged in the formation of $T_{i+1}$}\\
        C & \text{if $X_i$ is merged to form the cluster $C$ in  $T_{i+1}$}
    \end{cases}
$$ We find the largest $i$ such that $X_i$ does not have the separator flag. There are two cases to consider.
\begin{enumerate}
    \item Suppose that $x_{i+1} \notin Q$. Since $X_{i+1}$ has a separator flag, it must have a child cluster $Y$ that has the separator flag, and the cluster $X_i$ does not have the separator flag. So we know that there is a path from $x$ to some vertex $q\in Y\cap Q$ using $x_{i+1}$, so we call FindPathP2P($x$,$x_{i+1}$) and FindPath$'_i$($Y$,$x_{i+1}$), where we will describe below, to find the paths. The function then returns the union of the above two paths. 
    \item Suppose that $x_{i+1} \in Q$. Then we can call FindPathP2P($x$,$x_{i+1}$) directly to find the path.
\end{enumerate}
For the function FindPath'$_i$($Y$,$z$), the subscript $i$ indicates that the function is called on $T_i$. We assume $Y$ has $z$ as its boundary vertex and $Y$ has the separator flag. Suppose $Y$ is formed from the removal of vertex $y$ in $T_j$. Let $E$ be the (cluster) edge that connects $y$ and $z$ in $T_j$, if $E$ has the separator flag, we go to the lower level tree and call FindPath'$_{i-1}$($E$,$z$). Otherwise, we know $E$ does not have any vertex that belongs to $Q$. Let $E'$ be another child of $Y$ that has the separator flag, then we can call FindPathP2P($z,y$) and   FindPath'$_{i-1}$($y$,$E'$) because the union of the above two paths is the desired path. 

Thus the FindPath' function follows the following recursive (informal) relationship.
$$FindPath'_i = FindPath'_{i-1}+FindPathP2P$$

The depth of FindPathS2P$(C,x)$ query is $O(\log n)$ w.h.p. This is because the main function FindPathS2P only calls the function FindPath'$_i$ at most once for some $i$. For the function FindPath'$_i$, as shown above, the depth of the recursion is at most $i=O(\log n)$ w.h.p. In the meanwhile, we can run the FindPathP2P queries in parallel, and those operations have depth $O(\log n)$ w.h.p.

For the function FindPathS2P$(C,x)$, suppose that the returned path is a path between $x$ and a vertex $y\in C$. Then the work is $O(d\log n)$ w.h.p. where $d$ is the distance between $(x,y)$ in the original tree. This is because the work cost of FindPathS2P is the sum of all the work from FindPathP2P and $O(\log n)$ overhead cost. The work of FindPathP2P is proportional to its length, and the final path returned by the function is the union of the paths returned by all calls of FindPathP2P.

\fullOnly{
\section{Discussions}
We presented the first nearly work-efficient parallel DFS algorithm with sublinear depth for undirected graphs, concretely achieving $\tilde{O}(\sqrt{n})$ depth and $\tilde{O}(m)$ work. We see three interesting follow-up questions:

\begin{enumerate} 
\item Is there a (properly) work-efficient DFS algorithm for undirected graphs with depth sublinear in $n$? That is, can we get $O(m)$ work without any extra logarithmic factors, using, say, $\tilde{O}(\sqrt{n})$ depth?
\item Is there a DFS algorithm with $\tilde{O}(m)$ work and $\poly(\log n)$ depth for undirected graphs?
\item Is there a DFS algorithm with $\tilde{O}(m)$ work and sublinear depth---e.g., $\tilde{O}(\sqrt{n})$ or lower---depth for directed graphs? Notice that a work of Aggarwal, Anderson, and Kao~\cite{aggarwal1989parallel} gives a DFS algorithm with $\poly(n)$ work and $\poly(\log n)$ depth for directed graphs.
 \end{enumerate}
 }

 \section{Acknowledgment}
C. G. was supported by the European Research Council (ERC) under the European Unions Horizon 2020 research and innovation program (grant agreement No. 853109).

 \shortOnly{
 \balance
 }

\bibliographystyle{alpha}
\bibliography{ref}
\fullOnly{
\clearpage
\appendix

\section{Managing singularity cases in path merging}
In \Cref{pathreductionalgorithm}, we mentioned that two problematic special cases might arise in path merging:

\begin{enumerate}
    \item The discarded parts of $L^*$ caused $Q$ to no longer be a separator (\cref{discardproblem}).
    \item The number of matched paths is too small, meaning the number of paths in $P_1$ is less than $\frac{1}{12}k$ (\cref{fewpaths}).
\end{enumerate}

Aggarwal and Anderson~\cite{Aggarwal87} show that each of these singularity cases can be handled easily, and in either case, we can immediately find another separator consisting of less than $\frac{47}{48}k$ paths. Here, we recap their argument. 

If the number of paths in $P_1$ is less than $\frac{1}{12}k$, then \cref{fewpaths} shows that either $\hat{L}\cup P\cup S$ or $L\cup P\cup \hat{S}$ is a $\frac{23}{24}k$-path separator. So we need to check which one forms a separator. Checking whether $Q$ forms a separator can be done by computing the size of the largest connected component in $G-Q$. J\'{a}J\'{a}~\cite{JaJ92} provides such an algorithm with depth $O(\log n)$ and work $O(m\log n)$, where $m$ is the number of edges. In the other bad case where the number of paths in $P_1$ is greater or equal to $\frac{1}{12}k$ and $D\cup L^*$'s largest connected component has size more than $\frac{n}{2}$, \cref{discardproblem} shows that $L \cup \hat{S} \cup P$ is a $O(\frac{23}{24}k)$-path separator.

\paragraph{The Discarded Old Separator Problem}
\label{app:singularities}
Consider the scenario where we match at least $\frac{1}{12}k$ of the paths, but $Q$ no longer remains a separator after the updating. We already know that the largest connected component of $D$ has size less than $\frac{n}{2}$ because $L\cup S$ is a separator. If the updated separator no longer remains a separator, it is because $D\cup L^*$ has a connected component with size larger than $\frac{n}{2}$. Now we prove the following lemma.
\begin{lemma}
\label{discardproblem}
Suppose that the number of paths in $P_1$ is greater or equal to $\frac{1}{12}k$. If the largest connected component of $D\cup L^*$ has size larger than $\frac{n}{2}$, then the largest connected component of $D\cup (S-\hat{S})$ has size less than $\frac{n}{2}$. This means that $L \cup \hat{S} \cup P$ is an $O(\frac{23}{24}k)$-path separator.
\end{lemma}
\begin{proof}
Since $D$ does not have a connected component larger than $\frac{n}{2}$, it must be that the connected component containing $L^*$ has size larger than $\frac{n}{2}$. By property $2$, there is no path from $L^*$ to $S-\hat{S}$, which means that they are in different components of the induced graph $D\cup L^* \cup (S-\hat{S})$. If the connected component containing $L^*$ has size greater than $\frac{n}{2}$, then the connected component containing $S-\hat{S}$ has size less than $\frac{n}{2}$. 
\end{proof}
In this case, we discard $S-\hat{S}$ instead of $L^*$. The remaining paths are $L,\hat{S}, P$. The first two have at most $\frac{1}{4}k$ paths. And $P = P_1\cup P_2$. Set $P_1$ has at most $\frac{1}{4}k$ paths, and $P_2$ has at most $\frac{1}{48}k$ paths. Thus,  $L,\hat{S}, P$ together have at most $\frac{37}{48}k$ paths.

\paragraph{Too Few Paths are Matched} 
\begin{lemma}
\label{fewpaths}
     Suppose that the number of paths in $P_1$ is less than $\frac{1}{12}k$. Then either $\hat{L}\cup P\cup S$ or $L\cup P\cup \hat{S}$ is an $O(\frac{23}{24}k)$-path separator.  
\end{lemma}

In the other scenarios, $P_1$ has less than $\frac{1}{12}k$ paths. Consider $D \cup (L-\hat{L}) \cup (S-\hat{S})$, by property 1 and the same argument as above, either $D \cup (L-\hat{L})$ or $D \cup (S-\hat{S})$ has the property that the largest connected component has size less than $\frac{n}{2}$. In the former case, the remaining paths in the separator are $\hat{L}, P, S$. $\hat{L}$ has size less than $\frac{5}{48}k$, $P$ has size less than $\frac{5}{48}k$, and $S$ has size less than $\frac{3}{4}k$, together they have at most $\frac{23}{24}k$ paths. 

In the latter case, the remaining paths are $L,P,\hat{S}$. Set $L$ has size at most $\frac{k}{4}$, $P$ has size at most $\frac{5k}{48}$, and $\hat{S}$ has size $\frac{k}{12}$. Together they have at most $\frac{21k}{48}$ paths.

\section{Missing proof of Lemma \ref{lem:decremental_graph}}

\label{sec:appendix_missing_proofs}

This section is dedicated to prove \cref{lem:decremental_graph}, which we restate below for convenience. The proof of \cref{lem:decremental_graph} follows in a straightforward manner from \cref{lem:list_dynamic_decremental}, which is also proven in this section. 

\decrementalgraph*
\begin{proof}
For each vertex $v$ in $G$, we maintain the data structure from \cref{lem:list_dynamic_decremental} where the initial array $a_v$ contains one entry for each neighbor of $v$.
Throughout the execution, we maintain the invariant that $u$'s entry in $v$'s array is active if and only if $u$ is active in $G$.
We additionally store an array $b$ with $m$ entries. The $i$-th entry in the array corresponds to edge $e_i = \{u,v\}$. It stores the index corresponding to $u$ in $v$'s array and the index corresponding to $v$ in $u$'s array.
It follows from the guarantees of \cref{lem:list_dynamic_decremental} that initialization can be done in $O((m+n)\log n)$. We first discuss how to answer $MakeInactive(\{i_1,i_2,\ldots,i_k\})$. The algorithm first computes for each vertex $u$ neighboring at least one node $v_{i_j}$ a set consisting of all indices $i$ such that the $i$-th neighbor in $u$'s list is $v_{i_j}$, for some $j \in [k]$. These sets can be computed in $O((k + \sum_{j=1}^k deg_G(v_{i_j})) \log n)$ work and $O(\log n)$ depth using the array $b$.  Then, MakeInactive is called on $u$'s data structure with an array containing the indices as input. It follows from \cref{lem:list_dynamic_decremental} that this can be done in $O((k + \sum_{j=1}^k deg_G(v_{i_j})) \log n)$ work and $O(\log n)$ depth.
$Query(i_1,i_2,\ldots,i_k,t)$ can simply be answered by invoking $Query(t)$ on the datastructure for $v_{i_j}$ for all $j \in [k]$ in parallel. 
\end{proof}

Consider a list that undergoes (batches of) deletion. We want to efficiently query an arbitrary subset of $i$ elements from the list.

\begin{lemma}
    \label{lem:list_dynamic_decremental}
    There is a data structure with the following guarantees: 
    The initial input is an array $a$ with $N$ elements where in the beginning we think of all $N$ elements being active. The data structure supports the following operations after initialization:
    \begin{itemize}
        \item \textbf{MakeInactive$(\{i_1, i_2, \ldots, i_k\})$} takes an array consisting of $k\geq1$ distinct indices between $1$ and $N$ and marks the $k$ corresponding elements in the array as inactive. The work is $O(k \log N)$, and the depth is $O(\log N)$.
        \item \textbf{Query($t$)} takes a natural number $t \geq 1$ as input and returns $\min(t,N_{active})$ distinct active elements from the array, stored in an array of size $\min(t,N_{active})$. Here, $N_{active}$ is the remaining number of active elements. This operation can be done in $O(t \log n)$ work and $O(\log n)$ depth. 
    \end{itemize}
    Initialization can be done in $O(N \log N)$ work and $O(\log N)$ depth.
\end{lemma}
\begin{proof}
    We first construct a rooted static perfectly balanced binary search tree where the elements in $a$ correspond to the leaf nodes. Furthermore, each leaf node is augmented with a flag indicating whether it is active or not. Every node $v$ of the tree is augmented with $N(v)$, which is the number of remaining active elements in $a$ in the subtree rooted in $v$. For an active element $s$ in the array $a$, its corresponding leaf node $v_s$ has $N(v_s) = 1$ and if $s$ is not active anymore, then $N(v_s)=0$. Moreover, interior nodes additionally store $N(v_{left})$ and $N(v_{right})$ where $v_{left}$ is its left child and $v_{right}$ is its right child. This standard construction can be done in work $O(N \log N)$ and depth $O(\log N)$. 

    For the operation MakeInactive$(\{i_1,i_2,\ldots,i_k\})$, the leaves corresponding to the $k$ indices receive the makeinactive instruction. Those leaves update their flag to indicate that these elements are now inactive. Then the corresponding changes for the $N(v)$'s are recursively propagated up the tree in $O(k \log N)$ work and $O(\log N)$ depth.

    For the operation Query($t$), the root node $v$ first receives Query($t$). The root node is also augmented with $N(v)$, the total number of remaining active elements. For the discussion below, we can assume without loss of generality that $t$ is no larger than the total number of remaining active elements.
    Let $t_{left} := \min(N(v_{left}),t)$ and $t_{right} := t - t_{left}$.
    If $t_{left} > 0$, then we recursively construct an array $a_{left}$ consisting of $t_{left}$ distinct remaining active elements in the subtree rooted at $v_{left}$. Similarly, if $t_{right} > 0$, then we recursively construct an array $a_{right}$ consisting of $t_{right}$ distinct remaining active elements in the subtree rooted at $v_{right}$.

    If $t_{left} = 0$, then the array $a_{right}$ is returned, if $t_{right} = 0$, then the array $a_{left}$ is returned. If $t_{left},t_{right} > 0$, then one can concatenate the arrays $a_{left}$ and $a_{right}$ in $O(t)$ time. 
    As a base case, if this procedure reaches a leaf node, then an array of size $1$ with the corresponding (active) element is returned.

    Correctness, a work bound of $O(t \log N)$ and a depth bound of $O(\log N)$ follow by simple inductions.
\end{proof}

\section{Deterministic Parallel DFS}
\label{app:deterministic}
In this section, we provide a proof sketch that our randomized algorithm (\Cref{thm:main}) can be made deterministic with some slight modifications in the data structures and algorithms developed by Acar et al. in \cite{acar2020parallel} and \cite{acar2019parallel}, and at the expense of increasing the depth and work bounds by only logarithmic factors.

\begin{lemma}
    There is a deterministic algorithm that constructs a depth-first tree in $O(\sqrt{n}\poly\log n)$ depth and $O(m\poly\log n)$ work.
\end{lemma}
\begin{proof}[Proof Sketch]
    We first observe that new techniques developed in this paper to prove \Cref{thm:main} do not involve randomness. The only lemmas in our paper that use randomness are \cref{oldhdtlemma} and \cref{oldrclemma}. \cref{oldhdtlemma} uses results from the parallelized connectivity data structure~\cite{acar2019parallel}. \Cref{oldrclemma} uses results from the parallelized rake and compress tree~\cite{acar2020parallel}. We next list the parts in these papers that involve randomness, and then argue that there are deterministic analogs for each of these parts, which increase the work and depth bounds by only a factor of $\poly\log n$.\footnote{What we present here is only a proof sketch, in the sense that we only discuss the necessary changes. A full and self-contained proof would require recalling essentially all parts of these papers.} In our list, we identify the randomized part of the prior work with (R1) to (R5) and describe the corresponding deterministic solutions in (D1) to (D5).
    
    \paragraph{Parallel rake and compress tree~\cite{acar2020parallel}}
    \begin{itemize}
        \item[(R1)] For a long path in the rake and compress tree, to efficiently perform the compress operation, we would like to efficiently find an independent set on a path such that the set contains at least a constant fraction of the vertices on the path w.h.p. Acar et al.\cite{acar2020parallel} achieve this by flipping a coin for each vertex of degree 2, and putting a degree $2$ vertex $v$ in the independent set iff it has a head coin and both of its neighbors have tail coils. 

        \item[(D1)] Instead of using coin tosses, we use the classic deterministic approach of Cole and Vishkin for computing a maximal independent set (MIS). On an $n$-vertices path, this algorithm works using $\tilde{O}(n)$ work and $O(\log^* n)$ depth. We notice that an MIS contains at least a constant fraction of the vertices on the path. Thus, along with the rake operation, we can show that a constant fraction of the vertices is removed in each stage of the rake and compress tree, and thus the rake and compress again terminates in $O(\log n)$ iterations.
        
        We need some closer inspection in bounding the work in the dynamic applications of rake and compress.    One key argument in the paper \cite{acar2020parallel} by Acar et al. relies on the following: For a vertex $v$, the decision of whether to remove itself only depends on its neighbors and the leaf status of its neighbors. After some updates, they call a vertex \textit{affected} if its neighbors or the leaf status of its neighbors change. In $T_1$, with $k$ edge updates, at most $\{v_1 \codots v_{O(k)}\}$ vertices get affected. Affected vertices will infect new vertices, so they become affected in subsequent levels of rake and compress trees. All vertices in $T_i$ that have $v_j$ as their infection ``ancestor'' form a \textit{infection-tree}. A critical observation is that, in each infection tree, at most two vertices have unaffected neighbors. These are called boundary vertices. An uninfected node gets infected only through a boundary affected node. Moreover, each boundary vertex only infects new vertices if its degree in $T_i$ is 1 or 2. Hence, advancing to a higher level rake and compress tree, the number of affected vertices in a single infection-tree increases by at most a constant additive term. Thus only $O(\log n)$ vertices are affected due to a single edge update throughout $O(\log n)$ levels of the rake and compress tree. This means the additional work required for $k$ updates is $\tilde{O}(k)$.
        
        In our case, since the algorithm in \cite{coleandvishkin} has depth $O(\log^* n)$, the additional way for a vertex on a path to get affected is through some affected vertex on the same path in its $O(\log^* n)$ neighborhood. In fact, we declare all nodes inside the path within $O(\log^*n)$ distance of an affected node infected, in the next level. Hence, the number of affected vertices in a single \textit{infection-tree} increase at most by a $O(\log^* n)$ additive term in each level of the rake and compress tree. Like before, each infection tree still has at most $2$ boundary vertices. Thus only $O(\log n \log^* n)$ vertices are affected due to a single edge update throughout $O(\log n)$ levels of the rake and compress tree. So the additional work required for $k$ updates is still $\tilde{O}(k)$.
    \end{itemize}
\paragraph{Parallel connectivity data structure~\cite{acar2019parallel}}
    \begin{itemize}
        \item[(R2)] Acar et al.\cite{acar2019parallel} use Euler tours to represent the spanning tree structures, and these Euler tours are stored as skip lists. Skip lists support parallel operations like $k$ links and $k$ cuts, as well as parallel queries like $k$ connectivity in $O(k\log n)$ expected work and $O(\log n)$ depth w.h.p.
        \item[(D2)] Instead of using Euler tours to represent trees, we use the deterministic rake and compress tree structure sketched above. This deterministic data structure can still support parallel operations like $k$ links and $k$ cuts, as well as parallel queries like $k$ connectivity in $\tilde{O}(k)$ work and $\poly(\log n)$ depth. 
        \item[(R3)]  Acar et al.\cite{acar2019parallel} use parallel dictionaries are used to store the edges in the graph. Parallel dictionaries support $k$ insertions, $k$ deletions, and $k$ element look ups, for $k\in [1, \poly(n)]$, in $O(k)$ work and $O(\log^* n)$ time, with high probability.
        \item[(D3)] We can use an analog of the data structure developed in \cref{lem:list_dynamic_decremental} to store all potential edges in the graph, which are edges from the original graph $G$. This data structure supports $k$ insertions, $k$ deletions, and $k$ elements look up in $O(k\log n)$ work and $O(\log n)$ depth.
        \item[(R4)] Acar et al.\cite{acar2019parallel} use randomized semisorts which take $O(n)$ expected work and $O(\log n)$ depth, with high probability. 
        \item[(D4)] We can use a full deterministic sort instead of semisort. This takes $O(n\log n)$ work and $O(\log n)$ depth.  
        \item[(R5)] Acar et al.\cite{acar2019parallel} use Gazit's randomized connectivity (\cite{gazit}) algorithm to construct a spanning tree of a graph with $k$ edges, for $k\in [1, \poly(n)]$, in $O(k)$ expected work and $O(\log k)$ depth, with high probability.
        \item[(D5)] There are known deterministic parallel algorithms that construct a spanning tree of a graph with $k$ edges using $\tilde{O}(k)$ work and $\poly(\log n)$ depth\cite{JaJ92}. \qedhere
    \end{itemize}
\end{proof}
}

\end{document}